\documentclass[aps,prl,twocolumn,showpacs,longbibliography]{revtex4-1}

\usepackage[utf8]{inputenc}
\usepackage[english]{babel}
\usepackage[T1]{fontenc}
\usepackage{amsmath, amsfonts, amssymb, amsthm, mathrsfs, amssymb, appendix, bbm, bm, booktabs, graphicx,times,txfonts,color,subfigure}
\usepackage{hyperref}
\usepackage[pagecolor=white]{pagecolor}
\usepackage{xcolor}
\newcommand{\globalcolor}[1]{\color{#1}\global\let\default@color\current@color}

\newtheorem{theorem}{Theorem}

\newtheorem{corollary}{Corollary}

\newtheorem{proposition}[theorem]{Proposition}
\newtheorem{example}{Example}

\newcommand{\tinyspace}{\mspace{1mu}}

\newcommand{\rank}{\operatorname{rank}}

\newcommand{\im}{\operatorname{im}}

\newcommand{\defeq}{\stackrel{\smash{\textnormal{\tiny def}}}{=}}

\newcommand{\ket}[1]{\lvert\tinyspace #1 \tinyspace \rangle}
\newcommand{\bra}[1]{\langle\tinyspace #1 \tinyspace \rvert}

\newcommand{\setft}[1]{\mathrm{#1}}

\newcommand{\sr}[1]{\setft{SR}(#1)}
\newcommand{\sn}[1]{\setft{SN}(#1)}

\newcommand{\complex}{\mathbb{C}}

\newcommand\V{\mathcal{V}}

\renewcommand\H{\mathcal{H}}
\renewcommand\S{\mathcal{S}}
\newcommand{\w}{{\wedge}}

\newcommand\range{\setft{range}}
\newcommand\spn{\setft{span}}

\newcommand{\lip}{\langle}
\newcommand{\rip}{\rangle}

\newcommand{\ketbra}[2]{| \tinyspace #1 \tinyspace \rip\lip \tinyspace #2 \tinyspace |}

\begin{document}

\title{A Complete Hierarchy of Linear Systems for Certifying Quantum Entanglement of Subspaces}

\author{Nathaniel~Johnston}
\email{njohnston@mta.ca}
\affiliation{Department of Mathematics and Computer Science, Mount Allison University, Sackville, New Brunswick, Canada\\ Department of Mathematics \& Statistics, University of Guelph, Guelph, ON, Canada}

\author{Benjamin Lovitz}
\email{benjamin.lovitz@gmail.com}
\affiliation{Department of Mathematics, Northeastern University, Boston, Massachusetts, USA}

\author{Aravindan Vijayaraghavan}
\email{aravindv@northwestern.edu}
\affiliation{Department of Computer Science, Northwestern University, Evanston, Illinois, USA}

\begin{abstract}
    We introduce a hierarchy of linear systems for showing that a given subspace of pure quantum states is entangled (i.e., contains no product states). This hierarchy outperforms known methods already at the first level, and it is complete in the sense that every entangled subspace is shown to be so at some finite level of the hierarchy. It generalizes straightforwardly to the case of higher Schmidt rank, as well as the multipartite cases of completely and genuinely entangled subspaces. These hierarchies work extremely well in practice even in very large quantum systems, as they can be implemented via elementary linear algebra techniques rather than the semidefinite programming techniques that are required by previously-known hierarchies.
\end{abstract}


\maketitle


\color{black}
\section{Introduction}

Quantum entanglement is one of the central features of modern physics, and the problem of determining when entanglement is present in a quantum system is one of its most active research areas \cite{GT09,HHH09}. Of particular interest in this area is the problem of determining whether or not a given subspace is entangled. That is, the problem of determining whether or not every pure state in the subspace is entangled (i.e., not a product state) \cite{Par04,Bha06}.

In the bipartite setting of two quantum systems, one of the standard uses of certifying entanglement in subspaces is that any mixed quantum state supported on an entangled subspace is necessarily entangled \cite{Hor97,BDMSST99}, but numerous other applications have appeared in recent years. For example, entangled subspaces can be used to construct entanglement witnesses \cite{ATL11,CS14} and to perform quantum error correction \cite{GW07,HG20}. Further applications of this problem and its robust variants include determining the performance of QMA(2) protocols, computing the geometric measure of entanglement, and determining the ground-state energy of mean-field Hamiltonians as examples~\cite{HM10}. (For yet more applications, the reference~\cite{HM10} contains a compendium of 21 equivalent or closely related problems in quantum information and computer science!)


In the multipartite setting of three or more quantum systems, there are different notions of entanglement of a subspace. A completely entangled subspace in one containing no product states \cite{BDMSST99}, while a genuinely entangled subspace is one containing no states that are product across any bipartition (genuine entanglement is a stricter requirement than complete entanglement) \cite{MR18,AHB19}. Completely entangled subspaces are useful for locally discriminating pure quantum states \cite{Wal08,LJ21}, while genuinely entangled subspaces have been shown to have applications in quantum cryptography \cite{SS19}.

Determining whether or not a subspace is entangled is a difficult problem (see \cite{BUSS1999572} or \cite[Corollary~14]{HM10}, for example). To certify that a subspace is not entangled, it suffices to present a product vector in that subspace, but it is hard to actually find such a product vector in the first place. In the other direction, it is not known how to efficiently show that a given subspace \emph{is} entangled, even with the help of a certificate. To date, the only practical methods known for solving this problem work in very limited situations, such as when the subspace's dimension is smaller than the local dimensions \cite{LPS06,GR08,DRA21}, or when the dimensions are small enough that separability hierarchies based on semidefinite programming can be employed \cite{DPS04,NOP09,HNW17}.

We solve this problem by presenting a hierarchy of linear systems that can be used to certify that a given subspace is entangled. Our hierarchy is distinct from other hierarchies commonly used in quantum information theory; known semidefinite programming hierarchies are based on symmetric extensions and/or the sum of squares hierarchies \cite{DPS04}, while our hierarchy is based on Hilbert's projective Nullstellensatz from algebraic geometry \cite{harris2013algebraic}. As a result, our hierarchy terminates (i.e., detects every entangled subspace) at a finite level that depends only on the local dimensions; something that is known not to be possible for separability-based hierarchies like symmetric extensions \cite{Faw21}.

Our hierarchy works extremely well in practice, with even its first level being able to certify entanglement in subspaces that are much larger (quadratic in the local dimensions) than can be handled by other known techniques. The hierarchy also generalizes straightforwardly to $r$-entangled subspaces (i.e., subspaces in which every pure state has large Schmidt rank \cite{CMW08}), as well as to multipartite completely entangled subspaces and genuinely entangled subspaces. It also provides, as an immediate corollary, a new separability criterion that works well at detecting entanglement in low-rank mixed quantum states; even ones whose entanglement cannot be detected by the partial transpose \cite{Per96} (i.e., bound entangled states \cite{HHH98}). We provide MATLAB code that implements all of our methods \footnote{See \url{http://www.njohnston.ca/publications/entanglement-of-subspaces/} or the supplementary material of the arXiv version of this paper for MATLAB code\label{ref:suppmat}}.

\section{The First Level of the Hierarchy}

We use $\H_A$ and $\H_B$ to denote finite-dimensional complex Hilbert spaces (which can be thought of as $\complex^{d_A}$ and $\complex^{d_B}$) of dimension $d_A$ and $d_B$, respectively. A pure state $\ket{x} \in \H_A \otimes \H_B$ (i.e., a unit column vector) is said to be a \emph{product state} if it can be written in the form $\ket{x} = \ket{v} \otimes \ket{w}$ for some $\ket{v} \in \H_A$ and $\ket{w} \in \H_B$, and it is said to have \emph{Schmidt rank} $r$ (denoted by $\sr{\ket{x}} = r$) if it can be written as a linear combination (i.e., superposition) of $r$ product states but not fewer. A subspace of $\H_A \otimes \H_B$ is called \emph{$r$-entangled} (or just \emph{entangled} if $r = 1$) if every pure state in it has Schmidt rank $r+1$ or larger.

The starting point of our hierarchy for certifying that a given subspace is $r$-entangled is the observation that, for $\ket{x} \in \H_A \otimes \H_B$, we have $\sr{\ket{x}} \leq r$ if and only if
\begin{align}
    \big(P^{\wedge}_{A,r+1} \otimes P^{\wedge}_{B,r+1}\big)\big(\ket{x}^{\otimes (r+1)}\big) = 0,
\end{align}
where $P^{\wedge}_{A,r+1}$ is the projection onto the antisymmetric subspace of $\H_{A}^{\otimes(r+1)}$ (and similarly for the ``$B$'' subscripts). This is a classical result in algebraic geometry (see e.g., \cite{landsberg2012tensors}), and it has been used in a variety of contexts. For example, it appeared in a tensor decomposition algorithm in~\cite{cardoso1991super}, and a similar observation was made in \cite{TMG15}, where antisymmetric projections were used to create a semidefinite programming hierarchy for computing the Schmidt number of a mixed state.

For brevity, we define
\begin{align}\label{eq:phi_r_level1}
    \Phi_{r}^{1} \defeq P^{\wedge}_{A,r+1} \otimes P^{\wedge}_{B,r+1},
\end{align}
and for completeness, we formally state and prove the observation that we just made about $\Phi_r^1$:

\begin{proposition}\label{prop:phi_rk_rank_det}
    Suppose $\ket{x} \in \H_A \otimes \H_B$ and let $\Phi_r^1$ be the linear map from Equation~\eqref{eq:phi_r_level1}. Then $\sr{\ket{x}} \leq r$ if and only if $\Phi_r^1\big(\ket{x}^{\otimes (r+1)}\big) = 0$.
\end{proposition}

\begin{proof}
    Write $\ket{x}$ in its Schmidt decomposition $\ket{x} = \sum_{j=1}^s \lambda_j \ket{v_j} \otimes \ket{w_j}$ with $s = \sr{\ket{x}}$. Then
    \begin{widetext}
        \begin{align*}
            \Phi_r^1\big(\ket{x}^{\otimes (r+1)}\big) & = \big(P^{\wedge}_{A,r+1} \otimes P^{\wedge}_{B,r+1}\big)\big(\ket{x}^{\otimes (r+1)}\big) \\
            & = \sum_{j_1,\ldots,j_{r+1}=1}^s \lambda_{j_1}\cdots\lambda_{j_{r+1}}P^{\wedge}_{A,r+1}\Big(\ket{v_{j_1}} \otimes \cdots \otimes \ket{v_{j_{r+1}}}\Big) \otimes P^{\wedge}_{B,r+1}\Big(\ket{w_{j_1}} \otimes \cdots \otimes \ket{w_{j_{r+1}}}\Big).
        \end{align*}
    \end{widetext}
    If $s \leq r$ then $\{\ket{v_{j_1}},\ldots,\ket{v_{j_{r+1}}}\}$ is a set containing $r$ or fewer members, so $P^{\wedge}_{A,t}\big(\ket{v_{j_1}} \otimes \cdots \otimes \ket{v_{j_{r+1}}}\big) = 0$, so $\Phi_r^1\big(\ket{x}^{\otimes (r+1)}\big)$ = 0.
    
    On the other hand, if $s > r$ then for any $1 \leq \widetilde{j_1} < \cdots < \widetilde{j_{r+1}} \leq s$ we have
    \begin{align*}
        \big(\bra{v_{\widetilde{j_1}}} \otimes \cdots \otimes \bra{v_{\widetilde{j_{r+1}}}}\big) & P^{\wedge}_{A,r+1}\big(\ket{v_{j_1}} \otimes \cdots \otimes \ket{v_{j_{r+1}}}\big) \\
        = \big(\bra{w_{\widetilde{j_1}}} \otimes \cdots \otimes \bra{w_{\widetilde{j_{r+1}}}}\big) & P^{\wedge}_{B,r+1}\big(\ket{w_{j_1}} \otimes \cdots \otimes \ket{w_{j_{r+1}}}\big),
    \end{align*}
    and this quantity is non-zero if and only if $\big\{\widetilde{j_1},\ldots,\widetilde{j_{r+1}}\big\} = \{j_1,\ldots,j_{r+1}\}$. It follows that
    \[
        \Big(\bra{v_{\widetilde{j_1}}} \otimes \cdots \otimes \bra{v_{\widetilde{j_{r+1}}}} \otimes \bra{w_{\widetilde{j_1}}} \otimes \cdots \otimes \bra{w_{\widetilde{j_{r+1}}}}\Big)\Phi_r^1\big(\ket{x}^{\otimes (r+1)}\big)
    \]
    is non-zero (in fact, strictly positive). In particular, this means that $\Phi_r^1\big(\ket{x}^{\otimes (r+1)}\big) \neq 0$, completing the proof.
\end{proof}

The superscript ``$1$'' in the notation $\Phi_r^1$ refers to the fact that this map gives us the first level of our hierarchy for certifying that a subspace of $\H_A \otimes \H_B$ is $r$-entangled:

\begin{theorem}\label{thm:detector_level1}
    Let $\S \subseteq \H_A \otimes \H_B$ be a subspace with basis $\{\ket{x_1},\ldots,\ket{x_{d_S}}\}$. If the set
    \begin{align}\label{eq:li_level1}
        \left\{\Phi_{r}^{1}\big(\ket{x_{j_1}} \otimes \dots \otimes \ket{x_{j_{r+1}}}\big) : 1 \leq j_1 \leq \cdots \leq j_{r+1} \leq d_S \right\}
    \end{align}
    is linearly independent then $\S$ is $r$-entangled.
\end{theorem}

We do not yet prove this theorem, as it is a special case of the upcoming Theorem~\ref{thm:detector_levelk}, which we prove in the appendix.

Since the set~\eqref{eq:li_level1} consists of $\binom{d_S+r}{r+1}$ vectors, each living inside the $\binom{d_A}{r+1}\binom{d_B}{r+1}$-dimensional range of $\Phi_{r}^{1}$, Theorem~\ref{thm:detector_level1} can be implemented by determining whether or not a homogeneous $\binom{d_A}{r+1}\binom{d_B}{r+1} \times \binom{d_S+r}{r+1}$ linear system has a non-zero solution. Despite just being the first level of the hierarchy, this linear system can already certify $r$-entanglement of subspaces that are significantly larger than the local dimensions $d_A$ and $d_B$; a fact that we now illustrate with several examples and an additional proposition.

We say that a property holds for a \emph{generic $d_S$-dimensional subspace of $\H_A \otimes \H_B$} if it holds with probability one for a Haar-random $d_S$-dimensional subspace of $\H_A \otimes \H_B$ (see e.g.,~\cite[Definition~2.2]{Wal08} for the definition of a Haar-random subspace) \footnote{All of the genericity statements made here also hold more generally under the algebraic-geometric definition of generic presented in~\cite{LJ21}}. It is known that the maximum dimension of an $r$-entangled subspace is $(d_A-r)(d_B-r)$~\cite{CMW08}. For $r=1$, the following proposition shows that the first level of the hierarchy already certifies entanglement of a generic subspace of dimension up to a constant multiple of this maximum. This is surprising, given that the best-known algorithm in the worst case for determining whether a subspace is entangled or not runs in time exponential in $\sqrt{d_A}$ when $d_A=d_B$~\cite{barak2017quantum}.

\begin{proposition}\label{prop:generic}
    In the notation of Theorem~\ref{thm:detector_level1}, if $r = 1$ then the set \eqref{eq:li_level1} is linearly independent for a generic subspace of dimension $d_S < (d_A-1) (d_B-1) / 4$. 
\end{proposition}

We defer the proof of this proposition to the appendix.

Proposition~\ref{prop:generic} gives a sufficient condition for a generic subspace of dimension $d_S$ to be certified by the first level of our hierarchy. In the opposite direction, by just considering the size of the linear system that Theorem~\ref{thm:detector_level1} describes, we know that
\begin{align}\label{eq:dim_bound_level1}
    \binom{d_S+r}{r+1} \leq \binom{d_A}{r+1}\binom{d_B}{r+1}
\end{align}
is necessary.

The following pair of examples show that this bound is in fact tight in many cases, i.e., the first level of our hierarchy certifies entanglement of any subspace $\S$ for which $d_S$ satisfies Inequality~\eqref{eq:dim_bound_level1}.

\begin{example}\label{exx:dim4subspace}
    When $d_A = d_B = 4$ and $r = 1$, Inequality~\eqref{eq:dim_bound_level1} holds exactly when $d_S \leq 8$, so the largest subspace that we can hope to certify is entangled via Theorem~\ref{thm:detector_level1} has dimension $8$. It indeed works all the way up to dimension $8$, getting quite close to the maximum dimension of entangled subspaces of $(d_A-r)(d_B-r) = 9$ in this case.
    
    For example, following the construction of large entangled subspaces from \cite{CMW08}, consider the subspace
    \[
        \S = \mathrm{span}\big\{ \ket{x_1}, \ldots, \ket{x_8} \big\} \subset \H_A \otimes \H_B,
    \]
    where (here we omit normalization factors for brevity, and we use $\ket{j}$ to denote the $j$-th standard basis vector of $\H_A$ and $\H_B$)
    \begin{align*}
        \ket{x_1} & = \ket{0} \otimes \ket{0} + \ket{1} \otimes \ket{1} + \ket{2} \otimes \ket{2} + \ket{3} \otimes \ket{3}, \\
        \ket{x_2} & = \ket{0} \otimes \ket{0} - \ket{1} \otimes \ket{1} + \ket{2} \otimes \ket{2} - \ket{3} \otimes \ket{3}, \\
        \ket{x_3} & = \ket{0} \otimes \ket{1} + \ket{1} \otimes \ket{2} + \ket{2} \otimes \ket{3}, \\
        \ket{x_4} & = \ket{1} \otimes \ket{0} + \ket{2} \otimes \ket{1} + \ket{3} \otimes \ket{2}, \\
        \ket{x_5} & = \ket{0} \otimes \ket{1} + 2\ket{1} \otimes \ket{2} + 3\ket{2} \otimes \ket{3}, \\
        \ket{x_6} & = \ket{1} \otimes \ket{0} + 2\ket{2} \otimes \ket{1} + 3\ket{3} \otimes \ket{2}, \\
        \ket{x_7} & = \ket{0} \otimes \ket{2} + \ket{1} \otimes \ket{3}, \quad \text{and} \\
        \ket{x_8} & = \ket{2} \otimes \ket{0} + \ket{3} \otimes \ket{1}.
    \end{align*}
    To show that $\S$ is entangled, it suffices to solve the $\binom{d_A}{r+1}\binom{d_B}{r+1} \times \binom{d_S+r}{r+1} = 36 \times 36$ linear system described by Theorem~\ref{thm:detector_level1}. Doing so reveals that the set~\eqref{eq:li_level1} is indeed linearly independent, so $\S$ is entangled.
    
    Similarly, we generated $10^5$ Haar-random $8$-dimensional subspaces of $\H_A \otimes \H_B$, and Theorem~\ref{thm:detector_level1} detected their entanglement every single time (we will show in the upcoming Theorem~\ref{thm:detector_levelk} that this behavior is expected).
%
%
\end{example}

\begin{example}\label{exx:dim4subspacer2}
    When $d_A = d_B = 4$ and $r = 2$, Inequality~\eqref{eq:dim_bound_level1} holds exactly when $d_S \leq 3$, so the largest subspace that we can hope to certify is $2$-entangled via Theorem~\ref{thm:detector_level1} has dimension $3$. Many subspaces of this dimension are indeed certified, such as the span of the states $\ket{x_1}$, $\ket{x_3}$, and $\ket{x_4}$ from Example~\ref{exx:dim4subspace}. Performing this certification simply requires us to solve a $\binom{d_A}{r+1}\binom{d_B}{r+1} \times \binom{d_S+r}{r+1} = 16 \times 10$ linear system.
    
    Similarly, we generated $10^5$ Haar-random $3$-dimensional subspaces of $\H_A \otimes \H_B$, and Theorem~\ref{thm:detector_level1} detected their $2$-entanglement every single time.
\end{example}

Table~\ref{tab:level1_numerics} provides some numerics that show the maximum dimension of an $r$-entangled subspace that can be certified by Theorem~\ref{thm:detector_level1} (which, in all cases displayed, is equal to the largest value of $d_S$ for which Inequality~\eqref{eq:dim_bound_level1} holds) in various local dimensions, as well as the amount of time that it takes our code to certify such a subspace on a standard desktop computer. The subspaces that we checked to obtain these timings have a form that is similar to that of the subspace from Example~\ref{exx:dim4subspace}.

\begin{table}[!htb]
    \begin{center}
        \begin{tabular}{ @{ \ }c@{ \ \ } @{ \ \ }c@{ \ \ \ }c@{ \quad \quad }c@{ \ \ \ }c@{ \ } }
            \toprule
            \multicolumn{1}{c}{} & \multicolumn{2}{@{}c}{$r = 1$} & \multicolumn{2}{@{}c}{$r = 2$} \\ \cmidrule(r){2-3}\cmidrule{4-5}
            $d_A = d_B$ & max. $d_S$ & time & max. $d_S$ & time \\ \midrule
            $3$ & $3$ & $0.01$ s & $1$ & $0.03$ s \\
            $4$ & $8$ & $0.03$ s & $3$ & $0.19$ s \\
            $5$ & $13$ & $0.08$ s & $7$ & $0.65$ s \\
            $6$ & $20$ & $0.20$ s & $12$ & $2.38$ s \\
            $7$ & $29$ & $0.49$ s & $18$ & $8.17$ s \\
            $8$ & $39$ & $1.06$ s & $25$ & $27.46$ s \\
            $9$ & $50$ & $2.24$ s & $33$ & $1.78$ min \\
            $10$ & $63$ & $5.56$ s & $43$ & $14.62$ min \\ \bottomrule
        \end{tabular}
        \caption{The maximum dimension $d_S$ of a subspace of $\H_A \otimes \H_B$ that can be certified to be $r$-entangled by the first level of the hierarchy (i.e., Theorem~\ref{thm:detector_level1}), as well as the time required to do the certification, for small values of $d_A = d_B$ and $r$. In all cases shown here, the maximum dimension is the largest $d_S$ for which Inequality~\eqref{eq:dim_bound_level1} holds.}\label{tab:level1_numerics}
    \end{center}
%
\end{table}

\section{The Rest of the Hierarchy}

For an integer $k \geq 1$, the $k$-th level of our hierarchy is based on the following linear map acting on $(\H_A \otimes \H_B)^{\otimes (r+k)}$:
\begin{align}\label{eq:phi_rk_defn}
    \Phi_{r}^{k} \defeq \big(P^{\wedge}_{A,r+1} \otimes P^{\wedge}_{B,r+1} \otimes I_{AB,k-1}\big)P^{\vee}_{AB,r+k},
\end{align}
where $I_{AB,k-1}$ is the identity on $(\H_A \otimes \H_B)^{\otimes (k-1)}$ and $P^{\vee}_{AB,r+k}$ is the projection onto the $\binom{d_A d_B+r+k-1}{r+k}$-dimensional symmetric subspace of $(\H_A \otimes \H_B)^{\otimes(r+k)}$ (i.e., the symmetrization is performed between the $r+k$ copies of $\H_A \otimes \H_B$, but not between $\H_A$ and $\H_B$).

In the $k = 1$ case, $\Phi_r^k$ is exactly the same as the linear map $\Phi_r^1$ from Equation~\eqref{eq:phi_r_level1}, which can be seen by noting that $\range(P^{\wedge}_{A,r+1} \otimes P^{\wedge}_{B,r+1}) \subseteq \range(P^{\vee}_{AB,r+1})$. Theorem~\ref{thm:detector_level1} still works if $\Phi_r^1$ is replaced by $\Phi_r^k$, but we now furthermore get a converse that completely characterizes \emph{all} $r$-entangled subspaces:

\begin{theorem}\label{thm:detector_levelk}
    Let $\S \subseteq \H_A \otimes \H_B$ be a subspace with basis $\{\ket{x_1},\ldots,\ket{x_{d_S}}\}$. Then $\S$ is $r$-entangled if and only if there exists an integer $1 \leq k \leq (\max\{r,2\}+1)^{d_Ad_B}-r$ such that the set
    \begin{align}\label{eq:li_levelk}
        \left\{\Phi_{r}^{k}\big(\ket{x_{j_1}} \otimes \dots \otimes \ket{x_{j_{r+k}}}\big) : 1 \leq j_1 \leq \cdots \leq j_{r+k} \leq d_S \right\}
    \end{align}
    is linearly independent. Furthermore, if a subspace $\S$ is certified to be $r$-entangled at the $k$-th level of the hierarchy (i.e., if the set \eqref{eq:li_levelk} is linearly independent), then a generic $d_S$-dimensional subspace will be certified at the $k$-th level.
\end{theorem}

The proof of Theorem~\ref{thm:detector_levelk} is rather long and technical, so we defer it to the appendix. This theorem really does establish a \emph{hierarchy} for detecting $r$-entanglement in a subspace: if the set~\eqref{eq:li_levelk} is linearly independent for a particular value of $k$, then it is linearly independent for all larger values of $k$ as well.

While Theorem~\ref{thm:detector_levelk} only guarantees that the hierarchy detects \textit{all} $r$-entangled subspaces at its very high $k = (\max\{r,2\}+1)^{d_Ad_B} - r$ level, it is remarkable that a bound that does not depend on $\S$ exists at all (after all, no analogous bound can exist for semidefinite programming hierarchies for the separability problem \cite{Faw21}). Furthermore, the last sentence of the theorem allows us to show in practice that a much lower level (i.e., smaller value of $k$) suffices to detect most $r$-entangled subspaces, simply by finding a single $r$-entangled subspace of the maximal dimension $(d_A - r)(d_B - r)$ that is detected at that low level.

We have found such examples already at the $k = 2$ level of the hierarchy in many low-dimensional cases. Example~\ref{exx:subspace_33_level2} illustrates how such a certification at the 2nd level of the hierarchy works, and Table~\ref{tab:levelk_numerics} provides some numerics to show how long it takes this 2nd level of the hierarchy to certify $r$-entanglement of a maximum-dimensional subspace for some small values of the local dimensions and $r$.

\begin{example}\label{exx:subspace_33_level2}
    Suppose $d_A = d_B = 4$, $r = 1$, and $\ket{x_1},\dots,\ket{x_8}$ are as in Example~\ref{exx:dim4subspace}. If
    \[
        \ket{x_9} = \frac{1}{2}\big( \ket{0} \otimes \ket{0} + \ket{1} \otimes \ket{1} - \ket{2} \otimes \ket{2} - \ket{3} \otimes \ket{3} \big)
    \]
    then $\S := \mathrm{span}\big\{ \ket{x_1}, \ldots, \ket{x_9} \big\}$ cannot possibly be shown to be entangled by the first level of the hierarchy, since its dimension is too large. However, solving the $\binom{d_A}{r+1}\binom{d_B}{r+1}(d_Ad_B)^{k-1} \times \binom{d_S + r + k - 1}{r + k} = 576 \times 165$ linear system described by the second level of the hierarchy (i.e., Theorem~\ref{thm:detector_levelk} when $k = 2$) verifies that it is indeed entangled.
    
    As a bit of a side note, we observe that the number of rows in this linear system could be taken to be slightly less than $\binom{d_A}{r+1}\binom{d_B}{r+1}(d_Ad_B)^{k-1}$, since $\rank(\Phi_r^k)$ is actually smaller than $\rank(P^{\wedge}_{A,r+1} \otimes P^{\wedge}_{B,r+1} \otimes I_{AB,k-1}) = \binom{d_A}{r+1}\binom{d_B}{r+1}(d_Ad_B)^{k-1}$. However, indexing the range of $\Phi_r^k$ so as to take advantage of this (or even computing its rank exactly) is quite difficult.
\end{example}

\begin{table}[!htb]
    \begin{center}
        \begin{tabular}{ @{ \ }c@{ \ \ } @{ \ \ }c@{ \ \ \ }c@{ \quad \quad }c@{ \ \ \ }c@{ \ } }
            \toprule
            \multicolumn{1}{c}{} & \multicolumn{2}{@{}c}{$r = 1, k = 2$} & \multicolumn{2}{@{}c}{$r = 2, k = 2$} \\ \cmidrule(r){2-3}\cmidrule{4-5}
            $d_A = d_B$ & max. $d_S$ & time & max. $d_S$ & time \\ \midrule
            $3$ & $4$ & $0.11$ s & $1$ & $0.58$ s \\
            $4$ & $9$ & $0.47$ s & $4$ & $7.39$ s \\
            $5$ & $16$ & $1.38$ s & $9$ & $22.01$ s \\
            $6$ & $25$ & $8.04$ s & $16$ & $2.59$ min \\
            $7$ & $36$ & $48.42$ s & $25$ & $33.18$ min \\\bottomrule
        \end{tabular}
        \caption{A summary of how long it takes the 2nd level of the hierarchy (i.e., Theorem~\ref{thm:detector_levelk} with $k = 2$) to certify $r$-entanglement of a subspace of $\H_A \otimes \H_B$ with dimension $(d_A-r)^2$ (i.e., the maximum dimension), for small values of $d_A = d_B$ and $r$.}\label{tab:levelk_numerics}
    \end{center}
\end{table}

The size of the linear system described by Theorem~\ref{thm:detector_levelk} increases exponentially with $k$. However, it is also very sparse, so it can typically be solved even if it has hundreds of thousands of rows and columns.

\section{Certifying Schmidt Number of Low-Rank Mixed States}

The Schmidt number \cite{TH00} of a mixed quantum state $\rho$ acting on $\H_A \otimes \H_B$, denoted by $\sn{\rho}$, is the least integer $r$ such that $\rho$ is a convex combination of projectors onto Schmidt-rank-$r$ pure states from $\H_A \otimes \H_B$:
\begin{align}\label{eq:density_matrix_spectral}
    \rho = \sum_j p_j \ketbra{v_j}{v_j},
\end{align}
where $\{p_j\}$ is a probability distribution and each $\ket{v_j}$ has Schmidt rank at most $r$. If $\sn{\rho} = 1$ then $\rho$ is called separable, and it is called entangled otherwise \cite{Wer89}.

Determining whether a given mixed state is separable or entangled (or more generally, determining a state's Schmidt number) is a hard problem \cite{Gha10,Gur03}, so in practice numerous one-sided tests are used. One such test is the range criterion \cite{Hor97}, which says that if $\range(\rho)$ is not spanned by members of $\H_A \otimes \H_B$ with Schmidt rank at most $r$, then $\sn{\rho} \geq r+1$ (a fact that follows immediately from the decomposition~\eqref{eq:density_matrix_spectral} of $\rho$).

While the range criterion is simple to state and prove, actually making use of it is difficult, since it is difficult to show that a given subspace of $\H_A \otimes \H_B$ is not spanned by pure states with small Schmidt rank. Theorem~\ref{thm:detector_levelk} helps solve this problem, and immediately gives us the following result:

\begin{corollary}\label{cor:range_criterion_from_phir}
    Let $\rho$ be a mixed state acting on $\H_A \otimes \H_B$ with $d = \rank(\rho)$, and let $\{\ket{x_1},\ldots,\ket{x_d}\} \subseteq \H_A \otimes \H_B$ be a basis of $\range(\rho)$. If there exists an integer $k \geq 1$ such that
    \begin{align}\label{eq:schmidt_number_range}
        \left\{\Phi_{r}^{k}\big(\ket{x_{j_1}} \otimes \dots \otimes \ket{x_{j_{r+k}}}\big) : 1 \leq j_1 \leq \cdots \leq j_{r+k} \leq d \right\}
    \end{align}
    is linearly independent, then $\sn{\rho} \geq r+1$.
\end{corollary}

This corollary works best when applied to low-rank mixed states, and in particular we expect the first (i.e., $k = 1$) level of the hierarchy to detect most states' Schmidt number when $d_S = \rank(\rho)$ satisfies Inequality~\eqref{eq:dim_bound_level1}. Higher levels of the hierarchy allow for the certification of Schmidt number of higher-rank states, even ones whose entanglement cannot be detected by the celebrated positive partial transpose (PPT) criterion \cite{Per96}.

\begin{example}\label{ex:tiles_upb}
    Recall that if $U \subseteq \H_A \otimes \H_B$ is an unextendible product basis \cite{BDMSST99}, then the density matrix
    \[
        \rho_U := \frac{1}{d_Ad_B - |U|}\left( I - \sum_{\ket{v} \in U} \ketbra{v}{v} \right)
    \]
    is a PPT entangled state. For example, let $d_A = d_B = 3$ and consider the $5$-state ``Tiles'' UPB \cite{DMSST03} (here we omit normalization factors for brevity):
    \begin{align*}
        U_{\textup{tiles}} := \big\{ & \ket{0} \otimes (\ket{0} - \ket{1}), \ket{2} \otimes (\ket{1} - \ket{2}) \\
        & (\ket{0}-\ket{1}) \otimes \ket{2}, (\ket{1}-\ket{2}) \otimes \ket{0}, \\
        & (\ket{0} + \ket{1} + \ket{2}) \otimes (\ket{0} + \ket{1} + \ket{2}) \big\} \subset \H_A \otimes \H_B.
    \end{align*}
    The associated PPT entangled state $\rho_{U_{\textup{tiles}}}$ has $d = \rank(\rho_{U_{\textup{tiles}}}) = 4$, which is too high-rank for the $k = 1$ level of Corollary~\ref{cor:range_criterion_from_phir} to be able to detect entanglement in.
    
    However, we can apply the second level of that hierarchy by picking a basis of $\range(\rho)$ and then solving the $\binom{d_A}{r+1}\binom{d_B}{r+1}(d_Ad_B)^{k-1} \times \binom{d + r + k - 1}{r + k} = 81 \times 20$ linear system described by Corollary~\ref{cor:range_criterion_from_phir}. Doing so certifies (in about $0.1$ seconds) that $\rho_{U_{\textup{tiles}}}$ is entangled.
\end{example}

The above example is not a fluke: the map $\Phi_1^k$ detects entanglement in most low-rank states. For example, repeating the above example with the ``Tiles'' UPB replaced by any of the ``Pyramid'' \cite{BDMSST99}, ``QuadRes'', or ``GenTiles2'' UPBs \cite{DMSST03} yields the exact same conclusions: $\Phi_1^2$ detects the entanglement in the associated PPT entangled state.

The following example shows how the same method can be used to show that a low-rank mixed state isn't just entangled, but has Schmidt number strictly larger than $2$:

\begin{example}
    Let $d_A = d_B = 4$ and consider the mixed state
    \[
        \rho = \frac{1}{3}\sum_{j=1}^3 \ketbra{x_j}{x_j} \in \H_A \otimes \H_B,
    \]
    where (we again omit normalization factors for brevity)
    \begin{align*}
        \ket{x_1} & = \ket{0} \otimes \ket{0} + \ket{1} \otimes \ket{1} + \ket{2} \otimes \ket{2} + \ket{3} \otimes \ket{3}, \\
        \ket{x_2} & = \ket{0} \otimes \ket{1} + \ket{1} \otimes \ket{2} + \ket{2} \otimes \ket{3} + \ket{3} \otimes \ket{0}, \\
        \ket{x_3} & = \ket{0} \otimes \ket{2} + \ket{1} \otimes \ket{3} + \ket{2} \otimes \ket{0} - \ket{3} \otimes \ket{1}.
    \end{align*}

    The PPT criterion readily shows that $\rho$ is entangled (i.e., $\sn{\rho} \geq 2$), but we can say more by making use of $\Phi_2^1$. In particular, $\rho$ has rank $d = 3$, and solving the $\binom{m}{r+1}\binom{n}{r+1} \times \binom{d+r}{r+1} = 16 \times 4$ system of linear equations described by Corollary~\ref{cor:range_criterion_from_phir} shows that $\sn{\rho} \geq 3$.
\end{example}


\section{Multipartite Completely Entangled Subspaces}

Our hierarchy generalizes straightforwardly to the multipartite scenario (i.e., the tensor product of three or more Hilbert spaces). For example, a completely entangled subspace (CES) $\S$ of $\H_A \otimes \H_B \otimes \H_C$ is one containing no product vector (i.e., no vector of the form $\ket{u} \otimes \ket{v} \otimes \ket{w}$) \cite{Par04,Bha06}. We define $P^{\textup{CES}}_{ABC}$ to be the orthogonal projection onto
\begin{align}\label{eq:CES_map_k1}
    \w^2\H_A \otimes \w^2(\H_B\otimes \H_C)+\w^2(\H_A \otimes \H_B)\otimes \w^2 \H_C,
\end{align}
where $\w^2 \H$ denotes the antisymmetric (i.e., wedge) tensor product of two copies of $\H$. We emphasize that the subspace~\eqref{eq:CES_map_k1} is a sum of subspaces of $(\H_A \otimes \H_B \otimes \H_C)^{\otimes 2}$, but it is not a \emph{direct} sum of subspaces.

Since $\ket{x} \in \H_A \otimes \H_B \otimes \H_C$ is a product vector if and only if it is product across each of the $\H_A \otimes (\H_B \otimes \H_C)$ and $(\H_A \otimes \H_B) \otimes \H_C$ bipartitions, we have $\ket{x}$ being a product vector if and only if $P^{\textup{CES}}_{ABC}(\ket{x}^{\otimes 2}) = 0$. If we define the linear map
\begin{align}\label{eq:CES_map}
    \Phi_{\textup{CES}}^{k} \defeq \big(P^{\textup{CES}}_{ABC} \otimes I_{ABC,k-1}\big)P^{\vee}_{ABC,k+1},
\end{align}
then we have the following theorem that is directly analogous to the bipartite hierarchy provided by Theorem~\ref{thm:detector_levelk}:

\begin{theorem}\label{thm:ces_detector}
    Let $\S \subseteq \H_A \otimes \H_B \otimes \H_C$ be a subspace with basis $\{\ket{x_1},\ldots,\ket{x_{d_S}}\}$. Then $\S$ is completely entangled if and only if there exists an integer $1 \leq k \leq 3^{d_A d_B d_C}-r$ such that the set
    \begin{align}\label{eq:ces_detection}
        \left\{\Phi_{\textup{CES}}^{k}\big(\ket{x_{j_1}} \otimes \dots \otimes \ket{x_{j_{k+1}}}\big) : 1 \leq j_1 \leq \cdots \leq j_{k+1} \leq d_S \right\}
    \end{align}
    is linearly independent. Furthermore, if a subspace $\S$ is detected to be completely entangled at the $k$-th level of the hierarchy (i.e., if~\eqref{eq:ces_detection} is linearly independent), then a generic $d_S$-dimensional subspace will be detected at the $k$-th level.
\end{theorem}

The above theorem follows from Theorem~\ref{thm:nullstellensatz} via analogous arguments to those used in the proof of Theorem~\ref{thm:detector_levelk}, in the appendix.

\begin{example}\label{exam:ces_maximal}
    The largest possible dimension of a completely entangled subspace of $\H_A \otimes \H_B \otimes \H_C$ is $d_Ad_Bd_C - d_A - d_B - d_C + 2$, and one particular example of such a subspace is \cite{Bha06}
    \begin{align*}
        \S := \spn\Big\{ & \ket{i_A} \otimes \ket{i_B} \otimes \ket{i_C} - \ket{j_A} \otimes \ket{j_B} \otimes \ket{j_C} : \\
        & i_A + i_B + i_C = j_A + j_B + j_C \\
        & 0 \leq i_A, j_A < d_A, 0 \leq i_B, j_B < d_B, 0 \leq i_C, j_C < d_C \Big\}.
    \end{align*}
    Our method is able to certify this maximal-dimension CES for several small values of $d_A$, $d_B$, and $d_C$, as summarized in Table~\ref{tab:ces_numerics}.
\end{example}

\begin{table}[!htb]
    \begin{center}
        \begin{tabular}{ @{ \ }c@{ \ \ } @{ \ \ }c@{ \ \ }c@{ \ \ }c@{ \ } }
            \toprule
            $(d_A, d_B, d_C)$ & max. $d_S$ & level $k$ & time \\ \midrule
            $(2,2,2)$ & $4$ & $2$ & $0.12$ s \\
            $(2,2,3)$ & $7$ & $2$ & $0.30$ s \\
            $(2,2,4)$ & $10$ & $2$ & $0.67$ s \\
            $(2,2,5)$ & $13$ & $2$ & $1.21$ s \\
            $(2,2,6)$ & $16$ & $2$ & $3.47$ s \\
            $(2,2,7)$ & $19$ & $2$ & $6.05$ s \\
            $(2,2,8)$ & $22$ & $2$ & $18.90$ s \\
            $(2,2,9)$ & $25$ & $2$ & $38.40$ s \\
            $(2,3,3)$ & $12$ & $3$ & $19.58$ s \\
            $(2,3,4)$ & $17$ & $3$ & $8.24$ min \\
            $(2,3,5)$ & $22$ & $3$ & $2.50$ h \\
            $(3,3,3)$ & $20$ & $4$ & $14.68$ h \\\bottomrule
        \end{tabular}
        \caption{A summary of which level $k$ of the hierarchy from Theorem~\ref{thm:ces_detector} can be used to detect entanglement in the maximum-dimension completely entangled subspace of $\H_A \otimes \H_B \otimes \H_C$ from Example~\ref{exam:ces_maximal}, for small values of $d_A$, $d_B$, and $d_C$, as well as the computational time taken to do the certification.}\label{tab:ces_numerics}
    \end{center}
\end{table}

Theorem~\ref{thm:ces_detector} generalizes straightforwardly to the case of $p > 3$ parties using the fact that a multipartite vector $\ket{x}$ is product if and only if it is product across $p-1$ of its single-party bipartitions, and redefining $P^{\textup{CES}}_{ABC}$ accordingly. For example, if $p = 4$ then we would define $P^{\textup{CES}}_{ABCD}$ to be the orthogonal projection onto the (non-direct) sum
\begin{align*}
    \w^2\H_A \otimes \w^2(\H_B\otimes \H_C \otimes \H_D) & + \w^2 \H_B \otimes \w^2(\H_A \otimes \H_C \otimes \H_D) \\
    & + \w^2 \H_C \otimes \w^2(\H_A \otimes \H_B \otimes \H_D)
\end{align*}
and then define $\Phi_{\textup{CES}}^{k} = \big(P^{\textup{CES}}_{ABCD} \otimes I_{ABCD,k-1}\big)P^{\vee}_{ABCD,k+1}$. This map, if substituted into Theorem~\ref{thm:ces_detector}, provides a complete hierarchy for detecting completely entangled subspaces in $4$-party systems.

\section{Multipartite Genuinely Entangled Subspaces}

Another notion of multipartite entanglement of a subspace is that of a genuinely entangled subspace, which is a subspace in which no pure state is a product state across any bipartition \cite{MR18,AHB19}. Genuine entanglement is a stricter requirement than complete entanglement, since pure states can be separable across one or more bipartitions without being a product vector.

Our hierarchy can be applied directly to the case of genuinely entangled subspaces simply by applying Theorem~\ref{thm:detector_levelk} across every bipartition. For example, when trying to certify genuine entanglement of a subspace of $\H_A \otimes \H_B \otimes \H_C$, we consider the map $\Phi_1^k$ from Equation~\eqref{eq:phi_rk_defn} with respect to a particular bipartition of $\H_A \otimes \H_B \otimes \H_C$. That is, we define
\begin{align*}
    \Phi_{AB,C}^k \defeq \big(P^{\wedge}_{AB,2} \otimes P^{\wedge}_{C,2} \otimes I_{ABC,k-1}\big)P^{\vee}_{ABC,k+1},
\end{align*}
and similarly for $\Phi_{AC,B}^k$ and $\Phi_{BC,A}^k$. Theorem~\ref{thm:detector_levelk} then immediately implies the following corollary:

\begin{corollary}\label{cor:genuine_entanglement}
    Let $\S \subseteq \H_A \otimes \H_B \otimes \H_C$ be a subspace with basis $\{\ket{x_1},\ldots,\ket{x_{d_S}}\}$. Then $\S$ is genuinely entangled if and only if there exists an integer $1 \leq k \leq 3^{d_Ad_Bd_C}-1$ such that the sets
    \begin{align*}
        & \left\{\Phi_{AB,C}^{k}\big(\ket{x_{j_1}} \otimes \dots \otimes \ket{x_{j_{1+k}}}\big) : 1 \leq j_1 \leq \cdots \leq j_{1+k} \leq d_S \right\}, \\
        & \left\{\Phi_{AC,B}^{k}\big(\ket{x_{j_1}} \otimes \dots \otimes \ket{x_{j_{1+k}}}\big) : 1 \leq j_1 \leq \cdots \leq j_{1+k} \leq d_S \right\}, \\
        & \left\{\Phi_{BC,A}^{k}\big(\ket{x_{j_1}} \otimes \dots \otimes \ket{x_{j_{1+k}}}\big) : 1 \leq j_1 \leq \cdots \leq j_{1+k} \leq d_S \right\}
    \end{align*}
    are all linearly independent.
\end{corollary}

\begin{example}\label{exam:genuine_ent_c333}
    Let $d_A = d_B = d_C = 3$ and consider the $5$-dimensional genuinely entangled subspace of $\H_A \otimes \H_B \otimes \H_C$ that was introduced in \cite{AHB19} (see Proposition~2 of that paper, and the discussion afterwards). To certify that this subspace is genuinely entangled, we can apply the $k = 1$ case of Corollary~\ref{cor:genuine_entanglement}, which requires us to solve three $108 \times 15$ linear systems. Doing so verifies (in about 0.4 seconds) that it is indeed genuinely entangled.
\end{example}

The above corollary generalizes straightforwardly to any number of parties by similarly applying the map $\Phi_1^k$ from Equation~\eqref{eq:phi_rk_defn} to all $2^{p-1} - 1$ bipartitions of the $p$ parties.

\section{Conclusions}

We have introduced a hierarchy of systems of linear equations for certifying that a given subspace is entangled. This hierarchy is complete in the sense that every entangled subspace is certified to be so at a finite level that is independent of the subspace being checked. Since the hierarchy only depends on solving a linear system, it can be implemented much more easily, and it runs much quicker, than methods based on semidefinite programming. The hierarchy works extremely well in practice, with many entangled subspaces of interest being detected already at the first or second level, and it generalizes straightforwardly to higher Schmidt rank and the multipartite setting.

\section*{Acknowledgements}

\begin{acknowledgments}
N.J.\ was supported by NSERC Discovery Grant RGPIN-2022-04098. B.L. acknowledges that this material is based upon work supported by the National Science Foundation under Award No. DMS-2202782. A.V. was supported by the National Science Foundation under grants Grant No.~CCF-1652491, CCF 1934931.
\end{acknowledgments}

\bibliography{references}

\begin{thebibliography}{40}%
\makeatletter
\providecommand \@ifxundefined [1]{%
 \@ifx{#1\undefined}
}%
\providecommand \@ifnum [1]{%
 \ifnum #1\expandafter \@firstoftwo
 \else \expandafter \@secondoftwo
 \fi
}%
\providecommand \@ifx [1]{%
 \ifx #1\expandafter \@firstoftwo
 \else \expandafter \@secondoftwo
 \fi
}%
\providecommand \natexlab [1]{#1}%
\providecommand \enquote  [1]{``#1''}%
\providecommand \bibnamefont  [1]{#1}%
\providecommand \bibfnamefont [1]{#1}%
\providecommand \citenamefont [1]{#1}%
\providecommand \href@noop [0]{\@secondoftwo}%
\providecommand \href [0]{\begingroup \@sanitize@url \@href}%
\providecommand \@href[1]{\@@startlink{#1}\@@href}%
\providecommand \@@href[1]{\endgroup#1\@@endlink}%
\providecommand \@sanitize@url [0]{\catcode `\\12\catcode `\$12\catcode
  `\&12\catcode `\#12\catcode `\^12\catcode `\_12\catcode `\%12\relax}%
\providecommand \@@startlink[1]{}%
\providecommand \@@endlink[0]{}%
\providecommand \url  [0]{\begingroup\@sanitize@url \@url }%
\providecommand \@url [1]{\endgroup\@href {#1}{\urlprefix }}%
\providecommand \urlprefix  [0]{URL }%
\providecommand \Eprint [0]{\href }%
\providecommand \doibase [0]{http://dx.doi.org/}%
\providecommand \selectlanguage [0]{\@gobble}%
\providecommand \bibinfo  [0]{\@secondoftwo}%
\providecommand \bibfield  [0]{\@secondoftwo}%
\providecommand \translation [1]{[#1]}%
\providecommand \BibitemOpen [0]{}%
\providecommand \bibitemStop [0]{}%
\providecommand \bibitemNoStop [0]{.\EOS\space}%
\providecommand \EOS [0]{\spacefactor3000\relax}%
\providecommand \BibitemShut  [1]{\csname bibitem#1\endcsname}%
\let\auto@bib@innerbib\@empty
\bibitem [{\citenamefont {G\"{u}hne}\ and\ \citenamefont
  {T\'{o}th}(2009)}]{GT09}%
  \BibitemOpen
  \bibfield  {author} {\bibinfo {author} {\bibfnamefont {O.}~\bibnamefont
  {G\"{u}hne}}\ and\ \bibinfo {author} {\bibfnamefont {G.}~\bibnamefont
  {T\'{o}th}},\ }\bibfield  {title} {\enquote {\bibinfo {title} {Entanglement
  detection},}\ }\href@noop {} {\bibfield  {journal} {\bibinfo  {journal}
  {Physics Reports}\ }\textbf {\bibinfo {volume} {474}},\ \bibinfo {pages}
  {1--75} (\bibinfo {year} {2009})}\BibitemShut {NoStop}%
\bibitem [{\citenamefont {Horodecki}\ \emph {et~al.}(2009)\citenamefont
  {Horodecki}, \citenamefont {Horodecki}, \citenamefont {Horodecki},\ and\
  \citenamefont {Horodecki}}]{HHH09}%
  \BibitemOpen
  \bibfield  {author} {\bibinfo {author} {\bibfnamefont {R.}~\bibnamefont
  {Horodecki}}, \bibinfo {author} {\bibfnamefont {P.}~\bibnamefont
  {Horodecki}}, \bibinfo {author} {\bibfnamefont {M.}~\bibnamefont
  {Horodecki}}, \ and\ \bibinfo {author} {\bibfnamefont {K.}~\bibnamefont
  {Horodecki}},\ }\bibfield  {title} {\enquote {\bibinfo {title} {Quantum
  entanglement},}\ }\href@noop {} {\bibfield  {journal} {\bibinfo  {journal}
  {Reviews of Modern Physics}\ }\textbf {\bibinfo {volume} {81}},\ \bibinfo
  {pages} {865--942} (\bibinfo {year} {2009})}\BibitemShut {NoStop}%
\bibitem [{\citenamefont {Parthasarathy}(2004)}]{Par04}%
  \BibitemOpen
  \bibfield  {author} {\bibinfo {author} {\bibfnamefont {K.~R.}\ \bibnamefont
  {Parthasarathy}},\ }\bibfield  {title} {\enquote {\bibinfo {title} {On the
  maximal dimension of a completely entangled subspace for finite level quantum
  systems},}\ }\href@noop {} {\bibfield  {journal} {\bibinfo  {journal} {Proc.
  Indian Acad. Sci. (Math. Sci.)}\ }\textbf {\bibinfo {volume} {114}},\
  \bibinfo {pages} {365--374} (\bibinfo {year} {2004})}\BibitemShut {NoStop}%
\bibitem [{\citenamefont {Bhat}(2006)}]{Bha06}%
  \BibitemOpen
  \bibfield  {author} {\bibinfo {author} {\bibfnamefont {B.~V.~R.}\
  \bibnamefont {Bhat}},\ }\bibfield  {title} {\enquote {\bibinfo {title} {A
  completely entangled subspace of maximal dimension},}\ }\href@noop {}
  {\bibfield  {journal} {\bibinfo  {journal} {International Journal of Quantum
  Information}\ }\textbf {\bibinfo {volume} {4}},\ \bibinfo {pages} {325--330}
  (\bibinfo {year} {2006})}\BibitemShut {NoStop}%
\bibitem [{\citenamefont {Horodecki}(1997)}]{Hor97}%
  \BibitemOpen
  \bibfield  {author} {\bibinfo {author} {\bibfnamefont {P.}~\bibnamefont
  {Horodecki}},\ }\bibfield  {title} {\enquote {\bibinfo {title} {Separability
  criterion and inseparable mixed states with positive partial
  transposition},}\ }\href@noop {} {\bibfield  {journal} {\bibinfo  {journal}
  {Physics Letters A}\ }\textbf {\bibinfo {volume} {232}},\ \bibinfo {pages}
  {333--339} (\bibinfo {year} {1997})}\BibitemShut {NoStop}%
\bibitem [{\citenamefont {Bennett}\ \emph {et~al.}(1999)\citenamefont
  {Bennett}, \citenamefont {DiVincenzo}, \citenamefont {Mor}, \citenamefont
  {Shor}, \citenamefont {Smolin},\ and\ \citenamefont {Terhal}}]{BDMSST99}%
  \BibitemOpen
  \bibfield  {author} {\bibinfo {author} {\bibfnamefont {C.~H.}\ \bibnamefont
  {Bennett}}, \bibinfo {author} {\bibfnamefont {D.~P.}\ \bibnamefont
  {DiVincenzo}}, \bibinfo {author} {\bibfnamefont {T.}~\bibnamefont {Mor}},
  \bibinfo {author} {\bibfnamefont {P.~W.}\ \bibnamefont {Shor}}, \bibinfo
  {author} {\bibfnamefont {J.~A.}\ \bibnamefont {Smolin}}, \ and\ \bibinfo
  {author} {\bibfnamefont {B.~M.}\ \bibnamefont {Terhal}},\ }\bibfield  {title}
  {\enquote {\bibinfo {title} {Unextendible product bases and bound
  entanglement},}\ }\href@noop {} {\bibfield  {journal} {\bibinfo  {journal}
  {Physical Review Letters}\ }\textbf {\bibinfo {volume} {82}},\ \bibinfo
  {pages} {5385--5388} (\bibinfo {year} {1999})}\BibitemShut {NoStop}%
\bibitem [{\citenamefont {Augusiak1}\ \emph {et~al.}(2011)\citenamefont
  {Augusiak1}, \citenamefont {Tura},\ and\ \citenamefont {Lewenstein}}]{ATL11}%
  \BibitemOpen
  \bibfield  {author} {\bibinfo {author} {\bibfnamefont {R.}~\bibnamefont
  {Augusiak1}}, \bibinfo {author} {\bibfnamefont {J.}~\bibnamefont {Tura}}, \
  and\ \bibinfo {author} {\bibfnamefont {M.}~\bibnamefont {Lewenstein}},\
  }\bibfield  {title} {\enquote {\bibinfo {title} {A note on the optimality of
  decomposable entanglement witnesses and completely entangled subspaces},}\
  }\href@noop {} {\bibfield  {journal} {\bibinfo  {journal} {Journal of Physics
  A: Mathematical and Theoretical}\ }\textbf {\bibinfo {volume} {44}},\
  \bibinfo {pages} {212001} (\bibinfo {year} {2011})}\BibitemShut {NoStop}%
\bibitem [{\citenamefont {Chru\'{s}ci\'{n}ski}\ and\ \citenamefont
  {Sarbicki}(2014)}]{CS14}%
  \BibitemOpen
  \bibfield  {author} {\bibinfo {author} {\bibfnamefont {D.}~\bibnamefont
  {Chru\'{s}ci\'{n}ski}}\ and\ \bibinfo {author} {\bibfnamefont
  {G.}~\bibnamefont {Sarbicki}},\ }\bibfield  {title} {\enquote {\bibinfo
  {title} {Entanglement witnesses: construction, analysis and
  classification},}\ }\href@noop {} {\bibfield  {journal} {\bibinfo  {journal}
  {Journal of Physics A: Mathematical and Theoretical}\ }\textbf {\bibinfo
  {volume} {47}},\ \bibinfo {pages} {483001} (\bibinfo {year}
  {2014})}\BibitemShut {NoStop}%
\bibitem [{\citenamefont {Gour}\ and\ \citenamefont {Wallach}(2007)}]{GW07}%
  \BibitemOpen
  \bibfield  {author} {\bibinfo {author} {\bibfnamefont {G.}~\bibnamefont
  {Gour}}\ and\ \bibinfo {author} {\bibfnamefont {N.~R.}\ \bibnamefont
  {Wallach}},\ }\bibfield  {title} {\enquote {\bibinfo {title} {Entanglement of
  subspaces and error-correcting codes},}\ }\href@noop {} {\bibfield  {journal}
  {\bibinfo  {journal} {Physical Review A}\ }\textbf {\bibinfo {volume} {76}},\
  \bibinfo {pages} {042309} (\bibinfo {year} {2007})}\BibitemShut {NoStop}%
\bibitem [{\citenamefont {Huber}\ and\ \citenamefont {Grassl}(2020)}]{HG20}%
  \BibitemOpen
  \bibfield  {author} {\bibinfo {author} {\bibfnamefont {F.}~\bibnamefont
  {Huber}}\ and\ \bibinfo {author} {\bibfnamefont {M.}~\bibnamefont {Grassl}},\
  }\bibfield  {title} {\enquote {\bibinfo {title} {Quantum codes of maximal
  distance and highly entangled subspaces},}\ }\href@noop {} {\bibfield
  {journal} {\bibinfo  {journal} {Quantum}\ }\textbf {\bibinfo {volume} {4}},\
  \bibinfo {pages} {284} (\bibinfo {year} {2020})}\BibitemShut {NoStop}%
\bibitem [{\citenamefont {Harrow}\ and\ \citenamefont
  {Montanaro}(2010)}]{HM10}%
  \BibitemOpen
  \bibfield  {author} {\bibinfo {author} {\bibfnamefont {Aram~W.}\ \bibnamefont
  {Harrow}}\ and\ \bibinfo {author} {\bibfnamefont {Ashley}\ \bibnamefont
  {Montanaro}},\ }\bibfield  {title} {\enquote {\bibinfo {title} {An efficient
  test for product states with applications to quantum merlin-arthur games},}\
  }in\ \href {\doibase 10.1109/FOCS.2010.66} {\emph {\bibinfo {booktitle} {2010
  IEEE 51st Annual Symposium on Foundations of Computer Science}}}\ (\bibinfo
  {year} {2010})\ pp.\ \bibinfo {pages} {633--642}\BibitemShut {NoStop}%
\bibitem [{\citenamefont {Demianowicz}\ and\ \citenamefont
  {Augusiak}(2018)}]{MR18}%
  \BibitemOpen
  \bibfield  {author} {\bibinfo {author} {\bibfnamefont {M.}~\bibnamefont
  {Demianowicz}}\ and\ \bibinfo {author} {\bibfnamefont {R.}~\bibnamefont
  {Augusiak}},\ }\bibfield  {title} {\enquote {\bibinfo {title} {From
  unextendible product bases to genuinely entangled subspaces},}\ }\href@noop
  {} {\bibfield  {journal} {\bibinfo  {journal} {Physical Review A}\ }\textbf
  {\bibinfo {volume} {98}},\ \bibinfo {pages} {012313} (\bibinfo {year}
  {2018})}\BibitemShut {NoStop}%
\bibitem [{\citenamefont {Agrawal}\ \emph {et~al.}(2019)\citenamefont
  {Agrawal}, \citenamefont {Halder},\ and\ \citenamefont {Banik}}]{AHB19}%
  \BibitemOpen
  \bibfield  {author} {\bibinfo {author} {\bibfnamefont {S.}~\bibnamefont
  {Agrawal}}, \bibinfo {author} {\bibfnamefont {S.}~\bibnamefont {Halder}}, \
  and\ \bibinfo {author} {\bibfnamefont {M.}~\bibnamefont {Banik}},\ }\bibfield
   {title} {\enquote {\bibinfo {title} {Genuinely entangled subspace with
  all-encompassing distillable entanglement across every bipartition},}\
  }\href@noop {} {\bibfield  {journal} {\bibinfo  {journal} {Physical Review
  A}\ }\textbf {\bibinfo {volume} {99}},\ \bibinfo {pages} {032335} (\bibinfo
  {year} {2019})}\BibitemShut {NoStop}%
\bibitem [{\citenamefont {Walgate}\ and\ \citenamefont {Scott}(2008)}]{Wal08}%
  \BibitemOpen
  \bibfield  {author} {\bibinfo {author} {\bibfnamefont {J.}~\bibnamefont
  {Walgate}}\ and\ \bibinfo {author} {\bibfnamefont {A.~J.}\ \bibnamefont
  {Scott}},\ }\bibfield  {title} {\enquote {\bibinfo {title} {Generic local
  distinguishability and completely entangled subspaces},}\ }\href@noop {}
  {\bibfield  {journal} {\bibinfo  {journal} {Journal of Physics A:
  Mathematical and Theoretical}\ }\textbf {\bibinfo {volume} {41}},\ \bibinfo
  {pages} {375305} (\bibinfo {year} {2008})}\BibitemShut {NoStop}%
\bibitem [{\citenamefont {Lovitz}\ and\ \citenamefont {Johnston}(2022)}]{LJ21}%
  \BibitemOpen
  \bibfield  {author} {\bibinfo {author} {\bibfnamefont {B.}~\bibnamefont
  {Lovitz}}\ and\ \bibinfo {author} {\bibfnamefont {N.}~\bibnamefont
  {Johnston}},\ }\bibfield  {title} {\enquote {\bibinfo {title} {Entangled
  subspaces and generic local state discrimination with pre-shared
  entanglement},}\ }\href@noop {} {\bibfield  {journal} {\bibinfo  {journal}
  {Quantum}\ }\textbf {\bibinfo {volume} {6}},\ \bibinfo {pages} {760}
  (\bibinfo {year} {2022})}\BibitemShut {NoStop}%
\bibitem [{\citenamefont {Shenoy}\ and\ \citenamefont {Srikanth}(2019)}]{SS19}%
  \BibitemOpen
  \bibfield  {author} {\bibinfo {author} {\bibfnamefont {A.~H.}\ \bibnamefont
  {Shenoy}}\ and\ \bibinfo {author} {\bibfnamefont {R.}~\bibnamefont
  {Srikanth}},\ }\bibfield  {title} {\enquote {\bibinfo {title} {Maximally
  nonlocal subspaces},}\ }\href@noop {} {\bibfield  {journal} {\bibinfo
  {journal} {Journal of Physics A: Mathematical and Theoretical}\ }\textbf
  {\bibinfo {volume} {52}},\ \bibinfo {pages} {095302} (\bibinfo {year}
  {2019})}\BibitemShut {NoStop}%
\bibitem [{\citenamefont {Buss}\ \emph {et~al.}(1999)\citenamefont {Buss},
  \citenamefont {Frandsen},\ and\ \citenamefont {Shallit}}]{BUSS1999572}%
  \BibitemOpen
  \bibfield  {author} {\bibinfo {author} {\bibfnamefont {Jonathan~F}\
  \bibnamefont {Buss}}, \bibinfo {author} {\bibfnamefont {Gudmund~S}\
  \bibnamefont {Frandsen}}, \ and\ \bibinfo {author} {\bibfnamefont
  {Jeffrey~O}\ \bibnamefont {Shallit}},\ }\bibfield  {title} {\enquote
  {\bibinfo {title} {The computational complexity of some problems of linear
  algebra},}\ }\href {\doibase https://doi.org/10.1006/jcss.1998.1608}
  {\bibfield  {journal} {\bibinfo  {journal} {Journal of Computer and System
  Sciences}\ }\textbf {\bibinfo {volume} {58}},\ \bibinfo {pages} {572--596}
  (\bibinfo {year} {1999})}\BibitemShut {NoStop}%
\bibitem [{\citenamefont {Linden}\ \emph {et~al.}(2006)\citenamefont {Linden},
  \citenamefont {Popescu},\ and\ \citenamefont {Smolin}}]{LPS06}%
  \BibitemOpen
  \bibfield  {author} {\bibinfo {author} {\bibfnamefont {N.}~\bibnamefont
  {Linden}}, \bibinfo {author} {\bibfnamefont {S.}~\bibnamefont {Popescu}}, \
  and\ \bibinfo {author} {\bibfnamefont {J.~A.}\ \bibnamefont {Smolin}},\
  }\bibfield  {title} {\enquote {\bibinfo {title} {Entanglement of
  superpositions},}\ }\href@noop {} {\bibfield  {journal} {\bibinfo  {journal}
  {Physical Review Letters}\ }\textbf {\bibinfo {volume} {97}},\ \bibinfo
  {pages} {100502} (\bibinfo {year} {2006})}\BibitemShut {NoStop}%
\bibitem [{\citenamefont {Gour}\ and\ \citenamefont {Roy}(2008)}]{GR08}%
  \BibitemOpen
  \bibfield  {author} {\bibinfo {author} {\bibfnamefont {G.}~\bibnamefont
  {Gour}}\ and\ \bibinfo {author} {\bibfnamefont {A.}~\bibnamefont {Roy}},\
  }\bibfield  {title} {\enquote {\bibinfo {title} {Entanglement of subspaces in
  terms of entanglement of superpositions},}\ }\href@noop {} {\bibfield
  {journal} {\bibinfo  {journal} {Physical Review A}\ }\textbf {\bibinfo
  {volume} {77}},\ \bibinfo {pages} {012336} (\bibinfo {year}
  {2008})}\BibitemShut {NoStop}%
\bibitem [{\citenamefont {Demianowicz}\ \emph {et~al.}(2021)\citenamefont
  {Demianowicz}, \citenamefont {Rajchel-Mieldzio\'{c}},\ and\ \citenamefont
  {Augusiak}}]{DRA21}%
  \BibitemOpen
  \bibfield  {author} {\bibinfo {author} {\bibfnamefont {M.}~\bibnamefont
  {Demianowicz}}, \bibinfo {author} {\bibfnamefont {G.}~\bibnamefont
  {Rajchel-Mieldzio\'{c}}}, \ and\ \bibinfo {author} {\bibfnamefont
  {R.}~\bibnamefont {Augusiak}},\ }\bibfield  {title} {\enquote {\bibinfo
  {title} {Simple sufficient condition for subspace to be completely or
  genuinely entangled},}\ }\href@noop {} {\bibfield  {journal} {\bibinfo
  {journal} {New Journal of Physics}\ }\textbf {\bibinfo {volume} {23}},\
  \bibinfo {pages} {103016} (\bibinfo {year} {2021})}\BibitemShut {NoStop}%
\bibitem [{\citenamefont {Doherty}\ \emph {et~al.}(2004)\citenamefont
  {Doherty}, \citenamefont {Parrilo},\ and\ \citenamefont
  {Spedalieri}}]{DPS04}%
  \BibitemOpen
  \bibfield  {author} {\bibinfo {author} {\bibfnamefont {A.~C.}\ \bibnamefont
  {Doherty}}, \bibinfo {author} {\bibfnamefont {P.~A.}\ \bibnamefont
  {Parrilo}}, \ and\ \bibinfo {author} {\bibfnamefont {F.~M.}\ \bibnamefont
  {Spedalieri}},\ }\bibfield  {title} {\enquote {\bibinfo {title} {A complete
  family of separability criteria},}\ }\href@noop {} {\bibfield  {journal}
  {\bibinfo  {journal} {Physical Review A}\ }\textbf {\bibinfo {volume} {69}},\
  \bibinfo {pages} {022308} (\bibinfo {year} {2004})}\BibitemShut {NoStop}%
\bibitem [{\citenamefont {Navascu\'{e}s}\ \emph {et~al.}(2009)\citenamefont
  {Navascu\'{e}s}, \citenamefont {Owari},\ and\ \citenamefont
  {Plenio}}]{NOP09}%
  \BibitemOpen
  \bibfield  {author} {\bibinfo {author} {\bibfnamefont {M.}~\bibnamefont
  {Navascu\'{e}s}}, \bibinfo {author} {\bibfnamefont {M.}~\bibnamefont
  {Owari}}, \ and\ \bibinfo {author} {\bibfnamefont {M.~B.}\ \bibnamefont
  {Plenio}},\ }\bibfield  {title} {\enquote {\bibinfo {title} {Complete
  criterion for separability detection},}\ }\href@noop {} {\bibfield  {journal}
  {\bibinfo  {journal} {Physical Review Letters}\ }\textbf {\bibinfo {volume}
  {103}},\ \bibinfo {pages} {160404} (\bibinfo {year} {2009})}\BibitemShut
  {NoStop}%
\bibitem [{\citenamefont {Harrow}\ \emph {et~al.}(2017)\citenamefont {Harrow},
  \citenamefont {Natarajan},\ and\ \citenamefont {Wu}}]{HNW17}%
  \BibitemOpen
  \bibfield  {author} {\bibinfo {author} {\bibfnamefont {A.~W.}\ \bibnamefont
  {Harrow}}, \bibinfo {author} {\bibfnamefont {A.}~\bibnamefont {Natarajan}}, \
  and\ \bibinfo {author} {\bibfnamefont {X.}~\bibnamefont {Wu}},\ }\bibfield
  {title} {\enquote {\bibinfo {title} {An improved semidefinite programming
  hierarchy for testing entanglement},}\ }\href@noop {} {\bibfield  {journal}
  {\bibinfo  {journal} {Communications in Mathematical Physics}\ }\textbf
  {\bibinfo {volume} {352}},\ \bibinfo {pages} {881--904} (\bibinfo {year}
  {2017})}\BibitemShut {NoStop}%
\bibitem [{\citenamefont {Harris}(2013)}]{harris2013algebraic}%
  \BibitemOpen
  \bibfield  {author} {\bibinfo {author} {\bibfnamefont {J.}~\bibnamefont
  {Harris}},\ }\href {https://books.google.ca/books?id=U-UlBQAAQBAJ} {\emph
  {\bibinfo {title} {Algebraic Geometry: A First Course}}},\ Graduate Texts in
  Mathematics\ (\bibinfo  {publisher} {Springer New York},\ \bibinfo {year}
  {2013})\BibitemShut {NoStop}%
\bibitem [{\citenamefont {Fawzi}(2021)}]{Faw21}%
  \BibitemOpen
  \bibfield  {author} {\bibinfo {author} {\bibfnamefont {H.}~\bibnamefont
  {Fawzi}},\ }\bibfield  {title} {\enquote {\bibinfo {title} {The set of
  separable states has no finite semidefinite representation except in
  dimension $3 \times 2$},}\ }\href@noop {} {\bibfield  {journal} {\bibinfo
  {journal} {Communications in Mathematical Physics}\ }\textbf {\bibinfo
  {volume} {386}},\ \bibinfo {pages} {1319--1335} (\bibinfo {year}
  {2021})}\BibitemShut {NoStop}%
\bibitem [{\citenamefont {Cubitt}\ \emph {et~al.}(2008)\citenamefont {Cubitt},
  \citenamefont {Montanaro},\ and\ \citenamefont {Winter}}]{CMW08}%
  \BibitemOpen
  \bibfield  {author} {\bibinfo {author} {\bibfnamefont {T.~S.}\ \bibnamefont
  {Cubitt}}, \bibinfo {author} {\bibfnamefont {A.}~\bibnamefont {Montanaro}}, \
  and\ \bibinfo {author} {\bibfnamefont {A.}~\bibnamefont {Winter}},\
  }\bibfield  {title} {\enquote {\bibinfo {title} {On the dimension of
  subspaces with bounded {S}chmidt rank},}\ }\href@noop {} {\bibfield
  {journal} {\bibinfo  {journal} {Journal of Mathematical Physics}\ }\textbf
  {\bibinfo {volume} {49}},\ \bibinfo {pages} {022107} (\bibinfo {year}
  {2008})}\BibitemShut {NoStop}%
\bibitem [{\citenamefont {Peres}(1996)}]{Per96}%
  \BibitemOpen
  \bibfield  {author} {\bibinfo {author} {\bibfnamefont {A.}~\bibnamefont
  {Peres}},\ }\bibfield  {title} {\enquote {\bibinfo {title} {Separability
  criterion for density matrices},}\ }\href@noop {} {\bibfield  {journal}
  {\bibinfo  {journal} {Physical Review Letters}\ }\textbf {\bibinfo {volume}
  {77}},\ \bibinfo {pages} {1413--1415} (\bibinfo {year} {1996})}\BibitemShut
  {NoStop}%
\bibitem [{\citenamefont {Horodecki}\ \emph {et~al.}(1998)\citenamefont
  {Horodecki}, \citenamefont {Horodecki},\ and\ \citenamefont
  {Horodecki}}]{HHH98}%
  \BibitemOpen
  \bibfield  {author} {\bibinfo {author} {\bibfnamefont {M.}~\bibnamefont
  {Horodecki}}, \bibinfo {author} {\bibfnamefont {P.}~\bibnamefont
  {Horodecki}}, \ and\ \bibinfo {author} {\bibfnamefont {R.}~\bibnamefont
  {Horodecki}},\ }\bibfield  {title} {\enquote {\bibinfo {title} {Mixed-state
  entanglement and distillation: Is there a ``bound'' entanglement in
  nature?}}\ }\href@noop {} {\bibfield  {journal} {\bibinfo  {journal}
  {Physical Review Letters}\ }\textbf {\bibinfo {volume} {80}},\ \bibinfo
  {pages} {5239--5242} (\bibinfo {year} {1998})}\BibitemShut {NoStop}%
\bibitem [{Note1()}]{Note1}%
  \BibitemOpen
  \bibinfo {note} {See \protect \url
  {http://www.njohnston.ca/publications/entanglement-of-subspaces/} or the
  supplementary material of the arXiv version of this paper for MATLAB
  code\label {ref:suppmat}}\BibitemShut {NoStop}%
\bibitem [{\citenamefont {Landsberg}(2012)}]{landsberg2012tensors}%
  \BibitemOpen
  \bibfield  {author} {\bibinfo {author} {\bibfnamefont {Joseph~M.}\
  \bibnamefont {Landsberg}},\ }\href@noop {} {\emph {\bibinfo {title} {Tensors:
  {G}eometry and {A}pplications}}},\ Graduate studies in mathematics\ (\bibinfo
   {publisher} {American Mathematical Society},\ \bibinfo {year}
  {2012})\BibitemShut {NoStop}%
\bibitem [{\citenamefont {Cardoso}(1991)}]{cardoso1991super}%
  \BibitemOpen
  \bibfield  {author} {\bibinfo {author} {\bibfnamefont {Jean-Fran{\c{c}}ois}\
  \bibnamefont {Cardoso}},\ }\bibfield  {title} {\enquote {\bibinfo {title}
  {Super-symmetric decomposition of the fourth-order cumulant tensor. blind
  identification of more sources than sensors.}}\ }in\ \href@noop {} {\emph
  {\bibinfo {booktitle} {ICASSP}}},\ Vol.~\bibinfo {volume} {91}\ (\bibinfo
  {organization} {Citeseer},\ \bibinfo {year} {1991})\ pp.\ \bibinfo {pages}
  {3109--3112}\BibitemShut {NoStop}%
\bibitem [{\citenamefont {T\'oth}\ \emph {et~al.}(2015)\citenamefont {T\'oth},
  \citenamefont {Moroder},\ and\ \citenamefont {G\"uhne}}]{TMG15}%
  \BibitemOpen
  \bibfield  {author} {\bibinfo {author} {\bibfnamefont {G.}~\bibnamefont
  {T\'oth}}, \bibinfo {author} {\bibfnamefont {T.}~\bibnamefont {Moroder}}, \
  and\ \bibinfo {author} {\bibfnamefont {O.}~\bibnamefont {G\"uhne}},\
  }\bibfield  {title} {\enquote {\bibinfo {title} {Evaluating convex roof
  entanglement measures},}\ }\href@noop {} {\bibfield  {journal} {\bibinfo
  {journal} {Physical Review Letters}\ }\textbf {\bibinfo {volume} {114}},\
  \bibinfo {pages} {160501} (\bibinfo {year} {2015})}\BibitemShut {NoStop}%
\bibitem [{Note2()}]{Note2}%
  \BibitemOpen
  \bibinfo {note} {All of the genericity statements made here also hold more
  generally under the algebraic-geometric definition of generic presented
  in~\cite {LJ21}}\BibitemShut {NoStop}%
\bibitem [{\citenamefont {Barak}\ \emph {et~al.}(2017)\citenamefont {Barak},
  \citenamefont {Kothari},\ and\ \citenamefont {Steurer}}]{barak2017quantum}%
  \BibitemOpen
  \bibfield  {author} {\bibinfo {author} {\bibfnamefont {Boaz}\ \bibnamefont
  {Barak}}, \bibinfo {author} {\bibfnamefont {Pravesh~K}\ \bibnamefont
  {Kothari}}, \ and\ \bibinfo {author} {\bibfnamefont {David}\ \bibnamefont
  {Steurer}},\ }\bibfield  {title} {\enquote {\bibinfo {title} {Quantum
  entanglement, sum of squares, and the log rank conjecture},}\ }in\ \href@noop
  {} {\emph {\bibinfo {booktitle} {Proceedings of the 49th Annual ACM SIGACT
  Symposium on Theory of Computing}}}\ (\bibinfo {year} {2017})\ pp.\ \bibinfo
  {pages} {975--988}\BibitemShut {NoStop}%
\bibitem [{\citenamefont {Terhal}\ and\ \citenamefont
  {Horodecki}(2000)}]{TH00}%
  \BibitemOpen
  \bibfield  {author} {\bibinfo {author} {\bibfnamefont {B.~M.}\ \bibnamefont
  {Terhal}}\ and\ \bibinfo {author} {\bibfnamefont {P.}~\bibnamefont
  {Horodecki}},\ }\bibfield  {title} {\enquote {\bibinfo {title} {Schmidt
  number for density matrices},}\ }\href@noop {} {\bibfield  {journal}
  {\bibinfo  {journal} {Physical Review A}\ }\textbf {\bibinfo {volume} {61}},\
  \bibinfo {pages} {040301(R)} (\bibinfo {year} {2000})}\BibitemShut {NoStop}%
\bibitem [{\citenamefont {Werner}(1989)}]{Wer89}%
  \BibitemOpen
  \bibfield  {author} {\bibinfo {author} {\bibfnamefont {R.~F.}\ \bibnamefont
  {Werner}},\ }\bibfield  {title} {\enquote {\bibinfo {title} {Quantum states
  with {E}instein-{P}odolsky-{R}osen correlations admitting a hidden-variable
  model},}\ }\href@noop {} {\bibfield  {journal} {\bibinfo  {journal} {Physical
  Review A}\ }\textbf {\bibinfo {volume} {40}},\ \bibinfo {pages} {4277--4281}
  (\bibinfo {year} {1989})}\BibitemShut {NoStop}%
\bibitem [{\citenamefont {Gharibian}(2010)}]{Gha10}%
  \BibitemOpen
  \bibfield  {author} {\bibinfo {author} {\bibfnamefont {S.}~\bibnamefont
  {Gharibian}},\ }\bibfield  {title} {\enquote {\bibinfo {title} {Strong
  {NP}-hardness of the quantum separability problem},}\ }\href@noop {}
  {\bibfield  {journal} {\bibinfo  {journal} {Quantum Information and
  Computation}\ }\textbf {\bibinfo {volume} {10}},\ \bibinfo {pages} {343--360}
  (\bibinfo {year} {2010})}\BibitemShut {NoStop}%
\bibitem [{\citenamefont {Gurvits}(2003)}]{Gur03}%
  \BibitemOpen
  \bibfield  {author} {\bibinfo {author} {\bibfnamefont {L.}~\bibnamefont
  {Gurvits}},\ }\bibfield  {title} {\enquote {\bibinfo {title} {Classical
  deterministic complexity of {E}dmonds' problem and quantum entanglement},}\
  }in\ \href@noop {} {\emph {\bibinfo {booktitle} {Proceedings of the
  Thirty-Fifth Annual ACM Symposium on Theory of Computing}}}\ (\bibinfo {year}
  {2003})\ pp.\ \bibinfo {pages} {10--19}\BibitemShut {NoStop}%
\bibitem [{\citenamefont {DiVincenzo}\ \emph {et~al.}(2003)\citenamefont
  {DiVincenzo}, \citenamefont {Mor}, \citenamefont {Shor}, \citenamefont
  {Smolin},\ and\ \citenamefont {Terhal}}]{DMSST03}%
  \BibitemOpen
  \bibfield  {author} {\bibinfo {author} {\bibfnamefont {D.~P.}\ \bibnamefont
  {DiVincenzo}}, \bibinfo {author} {\bibfnamefont {T.}~\bibnamefont {Mor}},
  \bibinfo {author} {\bibfnamefont {P.~W.}\ \bibnamefont {Shor}}, \bibinfo
  {author} {\bibfnamefont {J.~A.}\ \bibnamefont {Smolin}}, \ and\ \bibinfo
  {author} {\bibfnamefont {B.~M.}\ \bibnamefont {Terhal}},\ }\bibfield  {title}
  {\enquote {\bibinfo {title} {Unextendible product bases, uncompletable
  product bases and bound entanglement},}\ }\href@noop {} {\bibfield  {journal}
  {\bibinfo  {journal} {Communications in Mathematical Physics}\ }\textbf
  {\bibinfo {volume} {238}},\ \bibinfo {pages} {379--410} (\bibinfo {year}
  {2003})}\BibitemShut {NoStop}%
\bibitem [{\citenamefont {Koll{\'a}r}(1988)}]{kollar1988sharp}%
  \BibitemOpen
  \bibfield  {author} {\bibinfo {author} {\bibfnamefont {J.}~\bibnamefont
  {Koll{\'a}r}},\ }\bibfield  {title} {\enquote {\bibinfo {title} {Sharp
  effective nullstellensatz},}\ }\href@noop {} {\bibfield  {journal} {\bibinfo
  {journal} {Journal of the American Mathematical Society}\ ,\ \bibinfo {pages}
  {963--975}} (\bibinfo {year} {1988})}\BibitemShut {NoStop}%
\end{thebibliography}%

\begin{appendices}
\section{Appendix: Proof of Theorem~\ref{thm:detector_levelk} and Proposition~\ref{prop:generic}}\label{app:proofs}

We now prove our main result---Theorem~\ref{thm:detector_levelk}, which Theorem~\ref{thm:detector_level1} occurs as a special case of. We require the following result, which essentially amounts to a translation of Hilbert's projective Nullstellensatz.

\begin{theorem}\label{thm:nullstellensatz}
    Let $r$ be a positive integer and let $\Psi_r^1 : \H_X^{\otimes (r+1)} \rightarrow \H_Y$ be a linear map that is invariant under all permutations of the $r+1$ copies of $\H_X$, i.e., $\Psi_r^1 P^{\vee}_{X,r+1}=\Psi_r^1$. Then the following statements are equivalent:
    \begin{enumerate}
        \item $\Psi_r^1(\ket{x}^{\otimes {(r+1)}}) \neq 0$ for all pure states $\ket{x} \in \H_X$.
        \item There exists a positive integer $1 \leq k \leq (\max\{r,2\}+1)^{d_X}-r$ for which $\range(P^{\vee}_{X,r+k})\cap \ker(\Psi_r^k)=\{0\}$,
        where $\Psi_r^k:= (\Psi_r^1 \otimes I_{X,k-1})P^{\vee}_{X,r+k}$.
    \end{enumerate}
\end{theorem}
\begin{proof}
    For $2 \Rightarrow 1$, if $\mathrm{range}(P^{\vee}_{X,r+k})\cap\ker(\Psi_r^k)=\{0\},$ then for all pure states $\ket{x}\in \H_X$ it holds that 
    \[
        0\neq \Psi_r^k(\ket{x}^{\otimes (r+k)})=\Psi_r^1(\ket{x}^{\otimes (r+1)})\otimes \ket{x}^{\otimes (k-1)},
    \]
    so $\Psi_r^1(\ket{x}^{\otimes (r+1)})\neq 0$. The converse $1 \Rightarrow 2$ is more difficult, and is obtained by translating Statement 1 to a statement about zeroes of homogeneous polynomials, invoking Hilbert's projective Nullstellensatz, and then translating back.
    
    In more details, first observe that the coordinates of $\Psi_r^1(\ket{x}^{\otimes (r+1)})$ as $\ket{x}$ ranges over the unit vectors in $\H_X$ can be written as homogeneous $d_X$-variate polynomials $p_1,\dots, p_{d_Y}$ in $\ket{x}$ of degree $r+1$, so Statement 1 is equivalent to there being no unit vector $\ket{x}$ (or equivalently, by scaling, no non-zero vector $\ket{x}$) for which $p_1(\ket{x})=\dots=p_{d_Y}(\ket{x})=0$. By Hilbert's projective Nullstellensatz and a degree bound due to Koll\'{a}r, this is equivalent to the existence of a positive integer $1 \leq k \leq (\max\{r,2\}+1)^{d_X}-r$ for which every degree $r+k$ monomial $x_{j_1}\cdots x_{j_{r+k}}$ can be written as a linear combination of the polynomials $q_{i_1,\dots, i_{k-1},j} \defeq x_{i_1}\cdots x_{i_{k-1}} p_{\ell}$, where $i_1,\dots, i_{k-1}$ range from $1$ to $d_X$, and $\ell$ ranges from $1$ to $d_Y$~\cite{kollar1988sharp,harris2013algebraic}.
    
    As with $\Psi_r^1(\ket{x}^{\otimes (r+1)})$, the coordinates of $\Psi_r^k(\ket{x}^{\otimes (r+k)})$ as $\ket{x}$ ranges over the unit vectors in $\H_X$ can be written as homogeneous $d_X$-variate polynomials of degree $r+k$. Direct calculation shows that these polynomials are precisely $q_{i_1,\dots, i_{k-1},j}$ (the identity map $I_{X,k-1}$ that appears in the definition of $\Psi_r^k$ produces the monomials $x_{i_1}\cdots x_{i_{k-1}}$). Since every monomial $x_{j_1}\cdots x_{j_{r+k}}$ can be written as a linear combination of the polynomials $q_{i_1,\dots, i_{k-1},j}$, there exists a linear map $\Xi : \H_Y \otimes \H_X^{\otimes (k-1)}\rightarrow \H_X^{\otimes (r+k)}$ for which $\Xi \circ \Psi_r^k(\ket{x}^{\otimes (r+k)})=\ket{x}^{\otimes (r+k)}$ for all $\ket{x}\in \H_X$. It follows that
    \[
        \ker(\Psi_r^k) \cap\spn\big\{\ket{x}^{\otimes (r+k)} : \ket{x} \in \H_X \big\}=\{0\}.
    \]
    This completes the proof, since $\spn\{\ket{x}^{\otimes (r+k)} : \ket{x} \in \H_X\} = \mathrm{range}(P^{\vee}_{X,r+k})$.
\end{proof}

In the following proof of Theorem~\ref{thm:detector_levelk}, we make use of Theorem~\ref{thm:nullstellensatz} in the special case where $\H_X = \H_A \otimes \H_B$ and $\H_Y = (\H_A \otimes \H_B)^{\otimes (r+1)}\oplus (\H_A \otimes \H_B)^{\otimes (r+1)}$.

\begin{proof}[Proof of Theorem~\ref{thm:detector_levelk}]
    Let $P_{\S}^{\perp}$ be the projection onto the orthogonal complement of $\S$. Then $\S$ is $r$-entangled if and only if there does not exist $\ket{x}\in \H_{A}\otimes \H_{B}$ for which $\Psi_r^1(\ket{x}^{\otimes (r+1)})=0$, where we define
    \[
        \Psi_r^1 : (\H_A \otimes \H_B)^{\otimes (r+1)} \rightarrow (\H_A \otimes \H_B)^{\otimes (r+1)}\oplus (\H_A \otimes \H_B)^{\otimes (r+1)}
    \]
    by
    \[
        \Psi_r^1\defeq\begin{bmatrix}\Phi_r^1 \\ (P_{\S}^{\perp} \otimes I_{AB, r})P^{\vee}_{AB,r+1}\end{bmatrix}.
    \]
    Indeed, by Proposition~\ref{prop:phi_rk_rank_det}, $\S$ is $r$-entangled if and only if for every $\ket{x} \in \H_A \otimes \H_B$ for which $P_{\S}^{\perp}\ket{x}= 0$, it holds that $\Phi_r^1(\ket{x}^{\otimes (r+1)})\neq 0$, which is easily seen to be equivalent to the above statement about $\Psi_r^1$.
    
    By Theorem~\ref{thm:nullstellensatz}, this is in turn equivalent to the existence of a positive integer $1 \leq k \leq (\max\{r,2\}+1)^{d_A d_B}-r$ for which $\mathrm{range}(P^{\vee}_{AB,r+k})\cap\ker(\Psi_r^k)=\{0\}$, where
    \[
        \Psi_r^k\defeq (\Psi_r^1 \otimes I_{AB,k-1})P^{\vee}_{AB,r+k}=\begin{bmatrix} \Phi_r^k \\ (P_{\S}^\perp \otimes I_{AB,r+k-1})P^{\vee}_{AB,r+k} \end{bmatrix}.
    \]
    Now, $\mathrm{range}(P^{\vee}_{AB,r+k})\cap\ker(\Psi_r^k)=\{0\}$ if and only if $\mathrm{range}(P^{\vee}_{AB,r+k}) \cap \mathrm{range} (P_{\S} \otimes I_{AB,r+k-1})\cap \ker(\Phi_r^k) = \{0\}$, where $P_{\S}$ denotes the projection onto $\S$. Observe that
    \begin{widetext}\begin{align*}
        \range\big(P^{\vee}_{AB,r+k}\big) \cap \range\big(P_{\S} \otimes I_{AB,r+k-1}\big) & = \range\big(P^{\vee}_{AB,r+k}\big) \cap \range\big(P_{\S}^{\otimes (r+k)}\big) \\
        & = \range\big(P^{\vee}_{AB,r+k}P_{\S}^{\otimes (r+k)}\big) \\
        & = \spn\big\{P^{\vee}_{AB,r+k}(\ket{x_{j_1}} \otimes \dots \otimes \ket{x_{j_{r+k}}}): 1\leq j_1,\dots, j_{r+k} \leq d_S\big\}\\
        & = \spn\big\{P^{\vee}_{AB,r+k}(\ket{x_{j_1}} \otimes \dots \otimes \ket{x_{j_{r+k}}}): 1\leq j_1\leq \dots \leq j_{r+k} \leq d_S\big\},
    \end{align*}\end{widetext}
    where the first line follows from permutation invariance, the second follows from the fact that the projections $P^{\vee}_{AB,r+k}$ and $P_{\S}^{\otimes (r+k)}$ commute, the third is clear, and the fourth follows from the fact that
    \[
        P^{\vee}_{AB,r+k}(\ket{x_{j_1}} \otimes \dots \otimes \ket{x_{j_{r+k}}})=P^{\vee}_{AB,r+k}(\ket{x_{j_{\sigma(1)}}} \otimes \dots \otimes \ket{x_{j_{\sigma(r+k)}}})
    \]
    for every permutation $\sigma$ of $\{1,2,\ldots,r+k\}$ (i.e., permutation invariance again). By permutation invariance of $\Phi_r^k$, $\S$ is $r$-entangled if and only if
    \begin{align*}
        \spn\big\{\ket{x_{j_1}} \otimes \dots \otimes \ket{x_{j_{r+k}}}: 1\leq j_1\leq \dots \leq j_{r+k} \leq d_S\big\} \cap \ker(\Phi_r^k) & \\
        = \{0\}, &
    \end{align*}
    i.e., the set in Equation~\eqref{eq:li_levelk} is linearly independent.
    
    For the statement beginning with ``Furthermore...," observe that linear independence of the set in Equation~\eqref{eq:li_levelk} is equivalent to the non-vanishing of some $\binom{d_S+r+k-1}{d_S-1}\times \binom{d_S+r+k-1}{d_S-1}$-minor of the matrix formed by taking the vectors in the set as columns. Since this determinant is a polynomial in the entries of $\ket{x_1},\dots, \ket{x_{d_S}},$ and any polynomial that is not identically zero vanishes on a set of Haar measure zero, this completes the proof.
\end{proof}

\newcommand{\calH}{\mathcal{H}}
\newcommand{\calS}{\mathcal{S}}
\newcommand{\calU}{\mathcal{U}}

\begin{proof}[Proof of Proposition~\ref{prop:generic}]

A generic subspace $\calS \subset \calH_A \otimes \calH_B$ of dimension $d_S$ can be chosen by picking $d_S$ generic vectors $\ket{x_1}, \dots, \ket{x_{d_S}} \in \calH_A \otimes \calH_B$ for the basis that spans $\calS$. Let $G=\{ P^{\vee}_{AB,2} (\ket{x_{j_1}} \otimes \ket{x_{j_{2}}}) : 1 \leq j_1 \leq j_{2} \leq d_S   \}$. We need to show that with probability $1$, 
\begin{equation}\label{eq:prop:1}
\ker(\Phi^1_1) \cap \spn(G)  = \{0\}.    
\end{equation}
We remark that the above condition is invariant under scaling of the vectors $\ket{x_1}, \dots, \ket{x_{d_S}}$. Hence, we will ignore the unit vector requirement for  $\ket{x_1}, \dots, \ket{x_{d_S}}$ (and all the vectors) for the purposes of this proof.

We now prove \eqref{eq:prop:1}. In the set $G$, the indices $j_1, j_2$ could be equal. We will partition the $\binom{d_S+1}{2}$ vectors in $G$ into subsets $G_{\text{eq}} = \{ P^{\vee}_{AB,2}(\ket{x_{j}} \otimes \ket{x_{j}}) : 1 \leq j \leq d_S   \}$ and $ G_{\text{neq}} = G \setminus G_{\text{eq}}$ has the terms with unequal indices. 
To establish \eqref{eq:prop:1}, it suffices to prove the following claim. 

\noindent {\bf Claim.}
\emph{With probability $1$ over the choice of $\ket{x_1}, \dots, \ket{x_{d_S}}$, we have for all $ 1 \le j_1 < j_2  \le d_S$, and all $ 1\le j \le d_S$,}  
\begin{align}
    P_{AB,2}^{\vee}(\ket{x_{j_1}}\otimes \ket{x_{j_2}}) &\notin \nonumber\\
    & \hspace{-1em} \text{span}\Big(\ker(\Phi^1_1) \cup G_{\text{neq}} \setminus \{P_{AB,2}^{\vee}(\ket{x_{j_1}} \otimes \ket{x_{j_2}})\} \Big), \label{eq:generic:helper0}
\end{align}
and
\begin{align}
    \ket{x_{j}}^{\otimes 2} &\notin \text{span}\Big(\ker(\Phi^1_1) \cup G_{\text{neq}} \cup G_{\text{eq}} \setminus \{\ket{x_{j}}^{\otimes 2}\} \Big). \label{eq:generic:helper1}
\end{align}  


To prove the claim, we first define the following subspaces. For each $i \in [d_S]$, let 
\begin{align*}
 \calU_i  &\coloneqq  \text{span}\Big\{ P^{\vee}_{AB,2}(\ket{x_i} \otimes \ket{z}) ~:~ \ket{z} \in \calH_A \otimes \calH_B \Big\}.
\end{align*}
Note that $\dim(\calU_i) \le d_A d_B$ and $\text{dim}(\ker(\Phi^1_1)) = (d_A d_B)^{2} - \binom{d_A}{2} \binom{d_B}{2}$. 


Consider $P^{\vee}_{AB,2}(\ket{x_{j_1}} \otimes \ket{x_{j_2}})$ for some $1 \le j_1 \le j_2 \le d_S$, and let $J^*=\{j_1\} \cup \{j_2\}$ representing the distinct indices involved.  
We observe that if $\calU_{-J^*} \coloneqq \bigcup_{i \in [d_S] \setminus J^*} \calU_i$, then 
\begin{align}
    \text{If } j_1 < j_2, ~~G_{\text{neq}} \setminus \{P^{\vee}_{AB,2}(\ket{x_{j_1}} \otimes \ket{x_{j_2}})\} &\subseteq \calU_{-J^*} ,\label{eq:generic:helper3}\\
    \text{else if} j_1 = j_2, ~~G_{\text{neq}} \cup G_{\text{eq}} \setminus \{\ket{x_{j_1}}^{\otimes 2} \} &\subseteq \calU_{-J^*} . \label{eq:generic:helper4}
\end{align} 
This is because when $j_1 < j_2$, every other vector in $G_{\text{neq}}$ involves at least one vector $\ket{x_j}$ with $j \in [d_S] \setminus J^*$. Hence \eqref{eq:generic:helper3} is true. Similarly when $j_1 = j_2$, we have \eqref{eq:generic:helper4} since every other vector in both $G_{\text{eq}}$ and $G_{\text{neq}}$ involves at least one vector $\ket{x_j}$ with $j \in [d_S] \setminus J^*$.  

The rest of the argument is the same for both \eqref{eq:generic:helper0} and \eqref{eq:generic:helper1}. 


Let $\V_{-J^*}:=\im(P^{\vee}_{AB,2}) \cap \text{ker}(\Phi^1_1) + \calU_{-J^*}.$ Then
\begin{align*}
    \dim(\V_{-J^*})&\leq \binom{d_A d_B + 1}{2} - \binom{d_A}{2}\cdot \binom{d_B}{2} + d_S (d_A d_B)\\
    &< \binom{d_A d_B + 1}{2},
\end{align*}
since $d_S \cdot (d_A d_B) <  \binom{d_A}{2} \binom{d_B}{2}$ by our assumption on $d_S$. It follows that $\V_{-J^*} \subsetneq \im(P^{\vee}_{AB,2})$. Hence $P^{\vee}_{AB,2}(\ket{x_{j_1}} \otimes \ket{x_{j_2}}) \notin \V_{-J^*}$ for a generic choice of $\ket{x_{j_1}},\ket{x_{j_2}}$ (note that $\V_{-J^*}$ does not depend on $\ket{x_{j_1}},\ket{x_{j_2}}$). This establishes both \eqref{eq:generic:helper0} and \eqref{eq:generic:helper1}, and completes the proof of Proposition~\ref{prop:generic}.
\end{proof}

\end{appendices}

\end{document}